\documentclass[letterpaper]{IEEEtran}
\usepackage[explicit]{titlesec}
\usepackage{graphicx}
\usepackage{verbatim}
\usepackage{epstopdf}
\usepackage{setspace}
\usepackage{subfigure} 
\usepackage{algorithmic}
\usepackage{amsmath}
\usepackage{amsthm}
\usepackage{mathrsfs}
\usepackage{extarrows}
\usepackage{cite}
\newtheorem{theorem}{Theorem}

\usepackage{bbm}
\usepackage{flushend}
\usepackage[bookmarks=false, colorlinks=true, citecolor=black]{hyperref}
\usepackage{amssymb}

\usepackage{colortbl}

\newcommand{\RNum}[1]{\uppercase\expandafter{\romannumeral #1\relax}}
\newtheorem{lemma}{Lemma}

\newtheorem{corollary}{Corollary}
\usepackage{stfloats}

\usepackage{color, soul}
\usepackage{booktabs}
\usepackage{tikz,xcolor}

\usepackage{balance}
\definecolor{lime}{HTML}{A6CE39}
\usepackage[left=0.68in,right=0.68in,top=0.67in,bottom=0.99in]{geometry}

\titlespacing{\section}{0pt}{1.2ex plus .0ex minus .0ex}{.3ex plus .0ex}
\titlespacing{\subsection}{0pt}{1.2ex plus .0ex minus .0ex}{.3ex plus .0ex}

\makeatletter
\DeclareRobustCommand{\orcidicon}{%
	\begin{tikzpicture}
	\draw[lime, fill=lime] (0,0) 
	circle [radius=0.16] 
	node[white] {{\fontfamily{qag}\selectfont \tiny ID}};    \draw[white, fill=white] (-0.0625,0.095) 
	circle [radius=0.007];    \end{tikzpicture}
	\hspace{-2mm}}
\foreach \x in {A, ..., Z}{%
	\expandafter\xdef\csname orcid\x\endcsname{\noexpand\href{https://orcid.org/\csname orcidauthor\x\endcsname}{\noexpand\orcidicon}}
}

\newcommand*\bigcdot{\mathpalette\bigcdot@{.5}}
\newcommand*\bigcdot@[2]{\mathbin{\vcenter{\hbox{\scalebox{#2}{$\m@th#1\bullet$}}}}}
\makeatother
\DeclareRobustCommand{\orcidicon}{%
	\begin{tikzpicture}
	\draw[lime, fill=lime] (0,0) 
	circle [radius=0.16] 
	node[white] {{\fontfamily{qag}\selectfont \tiny ID}};    \draw[white, fill=white] (-0.0625,0.095)
	circle [radius=0.007];    \end{tikzpicture}
	\hspace{-2mm}}
\foreach \x in {A, ..., Z}{%
	\expandafter\xdef\csname orcid\x\endcsname{\noexpand\href{https://orcid.org/\csname orcidauthor\x\endcsname}{\noexpand\orcidicon}}
}

\usepackage[linesnumbered,ruled,vlined]{algorithm2e}

\IEEEoverridecommandlockouts

\begin{document}
	\title{Error Floor of ML-Decoded Spinal Codes in the Finite Blocklength Regime}
	\author{Aimin Li\orcidA{}, Shaohua Wu\orcidB{}, \emph{Member, IEEE}, Xiaomeng Chen\orcidC{}, and Sumei Sun\orcidF{}, \emph{Fellow, IEEE\vspace{0em}}
		
		 

		\thanks{
			This work has been supported in part by the Guangdong Basic and Applied Basic Research Foundation under Grant no. 2022B1515120002, and in part by the National Natural Science Foundation of China under Grant no. 61871147, and in part by the Major Key Project of PCL under Grant no. PCL2024A01.  (\textit{Corresponding author: Shaohua Wu.})
			
		A. Li and X. Chen are with the Department of Electronic Engineering, Harbin Institute of Technology (Shenzhen), Shenzhen, China 518055. E-mail: liaimin@stu.hit.edu.cn; 23s052026@stu.hit.edu.cn. 
		
		S. Wu is with the Department of Electronic Engineering, Harbin Institute of Technology (Shenzhen), Shenzhen, China 518055, and also with the Peng Cheng Laboratory, 518055, China. E-mail:  hitwush@hit.edu.cn. 
		
		S. Sun is with the Institute for Infocomm Research, Agency for Science, Technology and Research (A*STAR), Singapore 138632. Email: sunsm@i2r.a-star.edu.sg.
	
	Copyright (c) 2025 IEEE. Personal use of this material is permitted. However, permission to use this material for any other purposes must be obtained from the IEEE by sending a request to pubs-permissions@ieee.org. 
		}
	}
	
	\maketitle
	\allowdisplaybreaks

	\begin{abstract}
		Spinal codes are a new family of \textit{capacity-achieving} rateless codes that have been shown to achieve better rate performance compared to Raptor codes, Strider codes, and \textit{rateless} Low-Density Parity-Check (LDPC) codes. This correspondence addresses the performance limitations of Spinal codes in the finite blocklength regime, uncovering an error floor phenomenon at high {Signal-to-Noise Ratios (SNRs)} under Maximum Likelihood (ML) Decoding. We develop an analytical expression to approximate the error floor and devise SNR thresholds at which the error floor initiates. Numerical results across Additive White Gaussian Noise (AWGN), Rayleigh, and Nakagami-m fading channels verify the accuracy of our analysis. The analysis and numerical results also show that transmitting more passes of symbols can lower the error floor but do not affect the SNR threshold, providing insights into the performance target, the operating SNR region, and the code design.
	\end{abstract}
	\begin{IEEEkeywords}
		Rateless Codes, Spinal Codes, Error Floor, AWGN channel, Fading Channel.
	\end{IEEEkeywords}
	
	\IEEEpeerreviewmaketitle
	
	\section{Introduction}\label{sectionI}
	\IEEEPARstart{S}pinal codes, first introduced in 2011 \cite{2011Spinal}, are the first \textit{rateless} codes that have been theoretically proven to achieve the \textit{Shannon capacity} over the \textit{Binary Symmetric Channel} (BSC) and the AWGN channel \cite{balakrishnan2012randomizing}. Experimentally, Spinal codes have been shown to achieve higher rate compared to Raptor codes \cite{Raptor}, Strider codes \cite{Strider}, and \textit{rateless} Low-Density Parity-Check (LDPC) codes \cite{LDPC} across a wide range of channel conditions and message sizes \cite{2012Spinal}. 
	
	Due to Spinal codes' \textit{rateless}, \textit{capacity-achieving} nature, and superior rate performance under short message lengths, Spinal codes have received extensive attention in the coding theory community. Key areas of investigation include the design of coding structures \cite{UEPspinal,yang2016two,compressive2019,li2021spinal}, development of high-efficiency decoding mechanisms \cite{8653979,yang2015low,DBLP:conf/iwcmc/BianLKD19,DBLP:conf/iccip/ZhangZB23}, applications in quantum information processing \cite{DBLP:journals/qip/WenLMLYH20}, optimizations of punctured codes \cite{xu2019optimized,li2020spinal, meng2022analysis}, and Polar-Spinal concatenation improved codes \cite{xu2018high,liang2020raptor,dong2017towards,cao2022polar}. These studies offer deeper insights into Spinal codes. \textcolor{black}{The excellent rate performance for short message length and the channel-adaptive rateless nature also make Spinal codes well-suited for deployment in 6G non-terrestrial networks (NTN) \cite{vaezi2022cellular}, which are expected to support a wide range of vehicle-to-vehicle (V2V) applications such as UAV networks \cite{wang2019multiple,pang2018trust}, deep space explorations \cite{wu2022cs}, Low Earth Orbit (LEO) satellite communications \cite{bian2018high}, and satellite laser communications \cite{iannucci2012no}.}
	
	The theoretical performance analysis of Spinal codes, however, is still in an early stage (See Table \ref{1} for an overview). In \cite{balakrishnan2012randomizing}, the \textit{asymptotic} performance of Spinal codes was investigated, proving that Spinal codes are \textit{capacity-achieving}. In \cite{UEPspinal} and \cite{8006920}, under the \textcolor{black}{Finite Block-Length (FBL)} regime, Yang, \textit{et al.} derived the bounds on the \textcolor{black}{Block Error Rate (BLER)} of Spinal codes over the Binary Erasure Channel (BEC) and the AWGN channel, respectively. In \cite{li2021new}, the upper bound on the BLER of Spinal codes over the AWGN channel was further tightened. In \cite{li2021spinal}, the upper bound on Spinal codes over the \textcolor{black}{Rayleigh} fading channel without \textcolor{black}{Channel State Information (CSI)} was derived. Then, in \cite{chen2023tight}, tight upper bounds over \textcolor{black}{Rayleigh, Rician, and Nakagami-m} fading channels under perfect CSI were explicitly derived. \cite{li2024new} further extended the result in \cite{chen2023tight} by deriving a \textit{unified} FBL BLER upper bound for general channels, achieving consistent tightness across a wide range of channel conditions.
	
	\begin{table}[t]\label{1}
		\centering
		\caption{A Summary of Theoretical Analysis of Spinal Codes}
		\begin{tabular}{cc}
			\toprule
			\textbf{Contribution} & \textbf{References}  \\
			\midrule
			Asymptotic Analysis & \cite{balakrishnan2012randomizing}  \\
			FBL BLER Bound over AWGN & \cite{UEPspinal,li2021new}  \\
			FBL BLER Bound over BEC & \cite{8006920} \\
			FBL BLER Bound over Fading Channels & \cite{meng2022analysis,li2021spinal,chen2023tight} \\
			Error Floor Analysis & \textbf{This work} \\
			\bottomrule
		\end{tabular}
	\end{table}
	
	While existing analyses have achieved precise approximations of Spinal codes' performance, the error floor phenomenon—where increasing SNR does not yield further improvements in error probability—has not been rigorously studied. Such a phenomenon has been documented in many other typical channel codes like Polar codes \cite{DBLP:journals/tit/MondelliHU16,DBLP:journals/tcom/EslamiP13,DBLP:conf/vtc/JuPCK21} and LDPC codes \cite{DBLP:journals/tcom/PsotaKP12,DBLP:journals/tit/AsvadiBA11,DBLP:journals/tcom/NeshaastegaranB21,DBLP:journals/icl/TaoCW22}, with several efforts aimed at lowering the error floor \cite{DBLP:conf/vtc/JuPCK21,DBLP:journals/tit/AsvadiBA11,DBLP:journals/icl/TaoCW22}.

	Motivated by this gap, this correspondence reveals that Spinal codes suffer from an error floor. We present an explicit expression for the error floor and analyze the SNR thresholds that trigger the error floors across AWGN, \textcolor{black}{Nakagami-m fading channels, and Rayleigh fading channels}. The analysis and numerical results indicate that transmitting additional passes of symbols can reduce the error floor but cannot change the SNR threshold, which offers insights into the SNR operating region, the performance target, and the design of codes.
	
	\section{Preliminaries of Spinal Codes}\label{section III}
	
	\subsection{Encoding of Spinal Codes}
	The encoding process of Spinal codes is shown in Fig. \ref{fig:spinal}:	

	\noindent 1) \textit{Segmentation}: Divide an $n$-bit message $\mathbf{M}$ into $k$-bit segments $\mathbf{m}_i \in \left\{0,1 \right\}^k$, where $i\in\{1,2,\cdots,n/k\}$.
	
	\noindent 2) \textit{Sequential Hashing}: The hash function sequentially generates $v$-bit spine values $\mathbf{s}_i \in {\{0,1\}}^v$, with
		\begin{equation} \label{eqhash}
		\mathbf{s}_i = h(\mathbf{s}_{i-1},\mathbf{m}_i),~i\in\{1,2,\cdots,n/k\},~\mathbf{s}_0 = \mathbf{0}^v \text{.} 
		\footnote{The initial spine value $\mathbf{s}_0$ is known to both the encoder and the decoder and is usually set as $\mathbf{s}_0=0$ without loss of generality.}
		\end{equation}
		\begin{figure}
			\centering
		\includegraphics[width=1\linewidth]{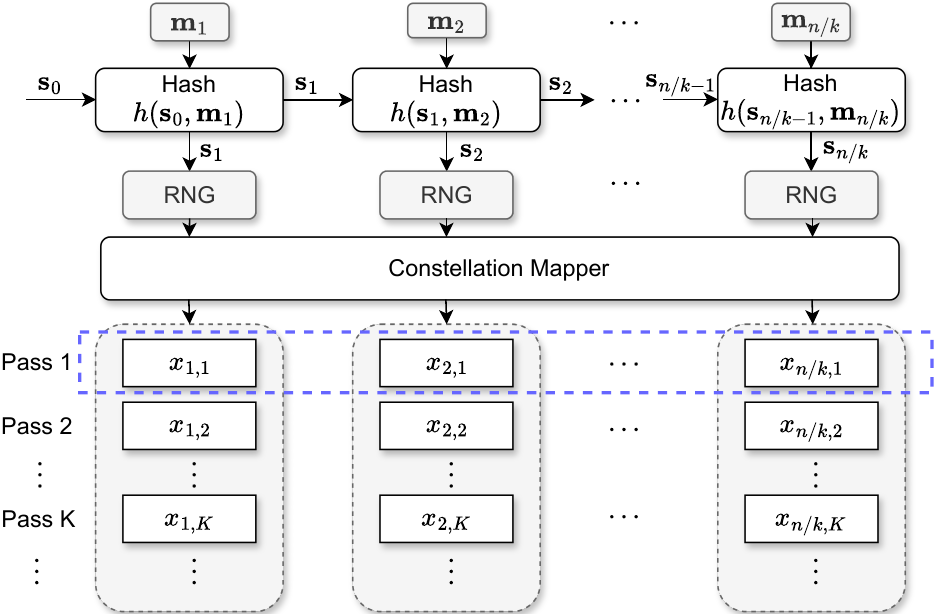}
			\caption{Encoding process of Spinal codes.}
			\label{fig:spinal}
			\vspace{-3mm}
		\end{figure}
\noindent 3) \textcolor{black}{\textit{Random Number Generator (RNG)}} and \textit{Constellation Mapping}: Each spine value $\mathbf{s}_i$ seeds an RNG to generate pseudo-random uniformly distributed symbols from $\{1,2,\cdots,2^c\}$, which is then mapped to the channel input set $\Psi$.
		\begin{equation}
		\mathrm{RNG:} ~\mathbf{s}_i \to \{{x}_{i,j}\}_{j \in \mathbb{N}^+}, x_{i,j}\in\Psi,
		\end{equation}
		where $j$ is the index of the pass\footnote{\textcolor{black}{Generally, the Spinal codes are transmitted pass by pass. The $j$-th pass of Spinal codes consists of $n/k$ coded symbols $\{x_{1,j},\cdots,x_{n/k,j}\}$. See Fig. \ref{fig:spinal} for the definition of a pass. }}.
	Fig. \ref{fig:mapping} illustrates two Quadrature Amplitude Modulation (QAM) examples of the constellation mapping set $\Psi$. 
		
	\subsection{ML Decoding of Spinal Codes}
	We consider ML decoding for Spinal codes in this correspondence. We denote $x_{i,j}(\mathbf{M})$ as the $j$-th pass symbol generated by $\mathbf{s}_i$ from $\mathbf{M}$. The received symbol at the receiver is denoted by $y_{i,j}=h_{i,j}\cdot x_{i,j}(\mathbf{M})+n_{i,j}$, where $h_{i,j}$ are independent and identically distributed (i.i.d.)  complex realization of the random variable $H$, which follows a distribution $f_H(h)$. The noise terms $n_{i,j}$ are i.i.d. complex Gaussian with $n_{i,j} \sim \mathcal{CN}(0,\sigma^2)$. We denote $L$ as the transmitted passes, and $\mathbf{M}^{\star}\in\{0,1\}^n$ as the output from the ML decoder. With the above notations, the ML decoding rule is given by
	\begin{equation} 
	\mathbf{M}^{\star}=\mathop{\arg\min}_{\mathbf{M}'\in\left\{0,1\right\}^n}\sum_{i=1}^{n/k}\sum_{j=1}^{L}({y}_{i,j}-{{x}}_{i,j}(\mathbf{{M}}'))^2.
	\end{equation}
	\section{Error Floor and SNR Threshold Analysis}\label{sectionIV}
	\subsection{Error Floor Analysis of Spinal Codes}
	The error floor analysis in this paper is built upon the BLER analysis in the FBL regime in our previous works~\cite{chen2023tight,li2024new}. In particular, we observed that the BLER does not vanish even as the AWGN noise variance $\sigma^2$ approaches zero, revealing the presence of a non-zero error floor. We begin by restating the BLER bound in the following lemma, whose proof can be found in~\cite{li2024new,chen2023tight}. This result serves as the foundation for the subsequent analysis of the error floor phenomenon.
	
	\begin{lemma}\label{the1} (\textcolor{black}{Restatement of \cite{li2024new,chen2023tight}})
		\label{coretheorem}
		Consider $(n,k,c,\Psi,L)$\footnote{For short-hand notations, we call Spinal codes with message length $n$, segmentation length $k$, modulation order $c$, channel input set $\Psi$, and transmitted passes $L$ as $(n,k,c,\Psi,L)$ Spinal codes.} Spinal codes transmitted over a complex fading channel with AWGN variance $\sigma^2$, the average BLER given perfect CSI under ML decoding for Spinal codes can be approximated by \textcolor{black}{a tight upper bound}
		\begin{equation}\label{eq4}
		\textcolor{black}{P_e \leq 1-\prod_{a=1}^{n/k} \left(1-\mathrm{min} \left\{ 1,\left(2^k-1\right)2^{n-ak} \cdot \mathscr{F} \left(L_a , \sigma \right) \right\}\right),}
		\end{equation}
		with $\mathscr{F} \left(L_a , \sigma \right)$ equals to
		\begin{equation} \label{allinone}
		\sum_{t=1}^N b_t\left(2^{-2c}\underbrace{\sum_{\beta_i,\beta_j \in \Psi}  \mathbb{E}_{H}\left[\exp\left\{\frac{-{|H(\beta_i-\beta_j)|}^2}{4 \sigma^2 \mathrm{sin}^2\theta_t}\right\}\right]}_{G(\Psi,\sigma^2,\theta_t,f_H(h))}\right)^{L_a},
		\end{equation}
		where $b_t = \frac{\theta_t-\theta_{t-1}}{\pi}$, $\theta_t$ is a set of quadrature angles satisfying $\theta_0=0$,  $\theta_N=\frac{\pi}{2}$ and $\theta_0 < \theta_1 < \cdots < \theta_{N} $. $N$ represents the number of $\theta$ values which enables the adjustment of accuracy. We define $L_a = L(n/k - a + 1)$ as the number of passes remaining at segment $a$. The expectation operator $\mathbb{E}_H[\cdot]$ is taken with respect to the probability density function $f_H(h)$ of the random channel coefficient $H$. The summation over $\beta_i, \beta_j \in \Psi$ enumerates all distinct pairs of constellation points in the input set $\Psi$ and is used to compute pairwise error contributions. The parameter $a$ is an indexing variable used in the product computation in \eqref{eq4}, iterating from $1$ to $n/k$.
	\end{lemma}
	\begin{figure}[t]
	\centering
	\includegraphics[width=1\linewidth]{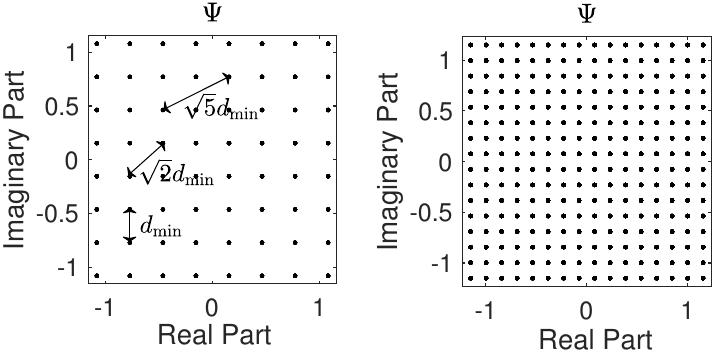}
	\caption{Examples of the constellation mapping set $\Psi$. The left panel is plotted under QAM modulation with $c=6$; The right panel is plotted under QAM modulation with $c=8$.}
	\label{fig:mapping}
	\vspace{-3mm}
\end{figure}
	Our error floor analysis is rooted in the key observation that $\mathscr{F}(L_a,0) > 0$. As a result, the error probability remains non-zero, regardless of the fact that SNR reaches infinity. Lemma \ref{l1} provides a mathematical demonstration of this phenomenon.
	
\begin{lemma}\label{l1}
	$\mathscr{F}(L_a,\sigma)$ is monotonically increasing in terms of $\sigma\ge0$, with the minimum value $\mathscr{F}(L_a,0)$ calculated by
	\begin{equation}
	\mathscr{F}(L_a,0)= 2^{-L_ac-1}.
	\end{equation}
\end{lemma}
\begin{proof}
	See Appendix \ref{appendixa}.
\end{proof}


	The error floor observed in Spinal codes stems from their random coding structure. This randomness can cause collisions, where different seeds, $\mathbf{s}_i$ and $\mathbf{s}_i'$, produce identical outputs, leading to $\beta_i = \beta_j$ as illustrated in \eqref{eq8}. Crucially, these collisions are an intrinsic feature of the random coding mechanism in Spinal codes and cannot be resolved by merely increasing the SNR.
	
	\begin{theorem}\label{the2}
	A $(n,k,c,\Psi,L)$ Spinal codes will suffer from an error floor $P_{EF}$, which is calculated by
	\begin{equation}
	P_{EF}=1-\prod_{a=1}^{n/k} \left(1-\mathrm{min} \left\{ 1,\left(2^k-1\right)2^{n-ak-L_ac-1} \right\}\right).
	\end{equation}
\end{theorem}
\begin{proof}
	Lemma \ref{l1} establishes that the function $\mathscr{F}(L_a, \sigma)$ is bounded below by $2^{-L_ac-1}$. This lower bound is attained when $\sigma = 0$, corresponding to the case where $\text{SNR} \rightarrow \infty$. Therefore, applying Lemma \ref{l1} in Lemma \ref{the1} yields the expression for the error floor of Spinal codes.
\end{proof}

\subsection{SNR Threshold of the Error Floor}
	Our analysis indicates that Spinal codes experience an error floor, where additional increases in SNR fail to lower the BLER. In practical applications, estimating the SNR threshold is crucial for efficient power allocation in Spinal codes. To address this need, we propose a method to approximate the SNR threshold for Spinal codes under QAM modulation\footnote{{We focus on QAM due to its analytical simplicity and wide adoption in practical systems. The proposed framework, however, is general for other modulations.}}. The main findings are presented in Theorem \ref{the3}.
	\begin{theorem}\label{the3}
		For an $(n,k,c,\Psi,L)$ Spinal codes where $c$ is a even, if the SNR $\gamma$ goes beyond a threshold with $\gamma\ge\gamma^{\text{th}}$, where $\gamma^{\text{th}}$ satisfies \begin{equation}\label{the3eq}
		4\cdot \mathbb{E}_{H}\left[\exp\left\{\frac{-{3|H|}^2\gamma^{\text{th}}}{2 (2^c-1) }\right\}\right]\ll 1,
		\end{equation}
		Then, the BLER of Spinal codes will suffer from an error floor.
	\end{theorem}
\begin{proof}
	See Appendix \ref{appendixb}.
\end{proof}
	
Based on Theorem \ref{the3}, we proceed to derive the SNR threshold for \textcolor{black}{Nakagami-m}, \textcolor{black}{Rayleigh}, and AWGN channels. These closed-form expressions are achieved by tailoring the distribution of the fading coefficient $H$ and explicitly calculating $\mathbb{E}_{H}\left[\exp\left\{\frac{-{3|H|}^2\gamma^{\text{th}}}{2 (2^c-1) }\right\}\right]$.
	\begin{corollary}\label{co1}
		(i) For \textcolor{black}{Nakagami-m fading channel with a normalized fading mean square} $\Omega=1$, the SNR threshold of the error floor is approximated by
		\begin{equation}
		\gamma^{\text{th}}\textcolor{black}{\approx}\frac{2m(2^c-1)}{3}\left(\textcolor{black}{\left(\frac{4}{x}\right)}^{1/m}-1\right).
		\end{equation}
		(ii) For \textcolor{black}{Rayleigh fading channel} with $\Omega=1$, the SNR threshold of the error floor is approximated by
		\begin{equation}
		\gamma^{\text{th}}\textcolor{black}{\approx}\frac{2(2^c-1)}{3}\left(\textcolor{black}{\left(\frac{4}{x}\right)}-1\right).
		\end{equation}
		(iii) For AWGN channel, the SNR threshold of the error floor is approximated by
		\begin{equation}
		\gamma^{\text{th}}\textcolor{black}{\approx}\frac{2(2^c-1)}{3}\ln\textcolor{black}{\left(\frac{4}{x}\right)}.
		\end{equation}
		where $x \ll 1$ is a small constant that can be chosen flexibly to adjust the precision of the approximation.
	\end{corollary}
\begin{proof}
	See Appendix \ref{appendixc}.
\end{proof}
\begin{figure}
	\centering	\includegraphics[width=1\linewidth]{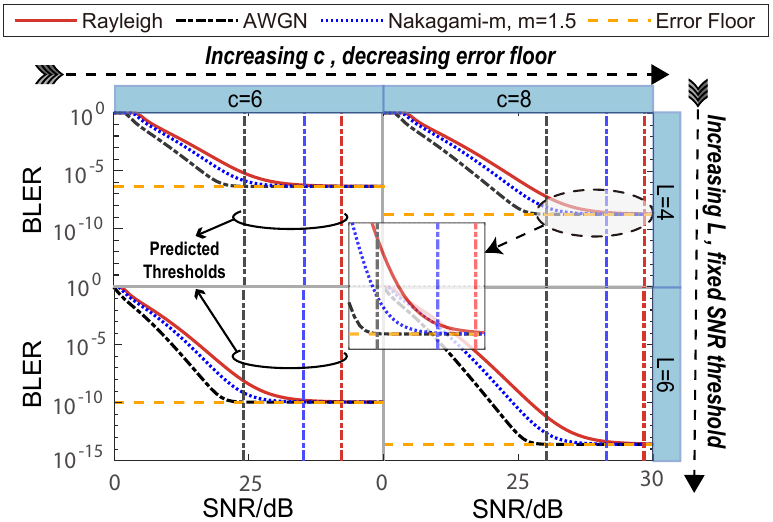}
	\caption{\textcolor{black}{BLER curves of Spinal codes where $n=32$ and $k=4$.}}
	\label{fig:simulation}
\end{figure}
\begin{figure}
	\centering	\includegraphics[width=1\linewidth]{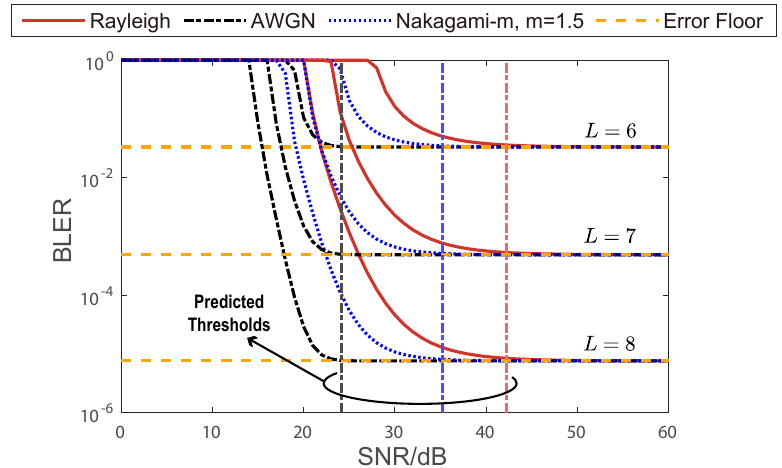}
	\caption{\textcolor{black}{BLER curves of Spinal codes where $n=256$, $k=32$, and $c=6$.}}
	\label{fig4:simulation}
\end{figure}
	\section{Numerical Verification}
Fig. \ref{fig:simulation} illustrates the error curves, error floor, and SNR thresholds over the AWGN, \textcolor{black}{Rayleigh}, and \textcolor{black}{Nakagami-m} fading channels with parameter $m=1.5$.

\subsection{Error Floor}
  From Fig. \ref{fig:simulation}, we observe that Spinal codes exhibit a consistent error floor across various channel types. This consistency indicates that the error floor is primarily attributed to the inherent random coding structure and the potential for collisions\footnote{\textcolor{black}{A collision occurs when distinct messages are encoded into the same codeword. For example, if $\mathbf{M}_1$ and $\mathbf{M}_2$ are two distinct messages such that their encoded codewords satisfy $\mathbf{X} = f(\mathbf{M}_1) = f(\mathbf{M}_2)$, where $f(\cdot)$ is the encoding function, then a collision has occurred. In random coding schemes like Spinal codes, the codeword $\mathbf{X}$ is generated randomly, resulting in a non-zero probability of collision and thus the error floor.}} within Spinal codes, rather than to external channel characteristics. The orange horizontal dashed line in each graph, as predicted by Theorem \ref{the2}, validates our theoretical analysis, demonstrating that the error floor is a fundamental limitation of Spinal codes, irrespective of the channel type.
 
 Additionally, as depicted in Fig. \ref{fig:simulation}, the error floor of Spinal codes decreases with increasing modulation order $c$ and the number of transmitted passes $L$, consistent with the predictions of Theorem \ref{the2}. However, since Spinal codes are \textit{rateless}, the number of transmitted passes $L$ continues to increase until successful decoding is achieved. Therefore, the modulation order parameter $c$ is an important variable that may help balance the trade-off between the error floor and the coding complexity of Spinal codes.
 
 \subsection{SNR Threshold}
 
 From Fig.~\ref{fig4:simulation}, we observe that the predicted SNR thresholds consistently mark the onset of the error floor across different types of fading channels and various parameter configurations. A comparative analysis of the first and second rows of subfigures clearly shows that the SNR thresholds are \textit{invariant} with respect to the number of transmitted passes, $L$, and are instead predominantly influenced by the modulation order, $c$, and the channel type. This behavior is in agreement with the theoretical insights established in Corollary~\ref{co1}. This observation highlights a key insight: for Spinal codes, increasing the number of symbol passes can effectively lower the error floor but does not shift the SNR threshold at which it begins. 
 
	\section{Conclusion}\label{sectionVI}In this correspondence, we studied the error floor in Spinal codes and determined that their random coding structure is the underlying cause. We derived explicit expressions for the error floor associated with Spinal codes and calculated approximate SNR thresholds under QAM modulation, indicating the onset of the error floor across \textcolor{black}{Rayleigh}, \textcolor{black}{Nakagami-m}, and AWGN channels. Our numerical results corroborated these findings. In addition, the analysis and numerical results revealed that, while transmitting additional passes of symbols can reduce the error floor, it cannot change the SNR threshold. This provides insights into the SNR working region and informs the design and transmission practices of codes. A promising direction for future work lies in developing collision-mitigation mechanisms that retain the rateless and scalable nature of Spinal codes. Such efforts may require new theoretical tools to model and control the trade-offs between encoding randomness, collision probability, and decoding complexity.

	\appendices
	\section{Proof of Lemma \ref{l1}}\label{appendixa}
	\begin{proof}
		The monotonically increasing property is trivial. Our focus is on solving $\mathscr{F}(L_a,0)$, which is equivalent to addressing $G\left(\Psi,0,\theta_t,f_H(h)\right)$ as specified in \eqref{allinone}. To facilitate this, we first rewrite the function $G\left(\Psi,\sigma^2,\theta_t,f_H(h)\right)$ as
		\begin{equation}
		\begin{aligned}
		\underset{\substack{\beta_i,\beta_j\in\Psi\\\beta_i=\beta_j}}{\sum}  \mathbb{E}_{H}\left[\exp\left\{\frac{-{|H(\beta_i-\beta_j)|}^2}{4 \sigma^2 \mathrm{sin}^2\theta_t}\right\}\right]+\\\underset{\substack{\beta_i,\beta_j\in\Psi\\\beta_i\ne\beta_j}}{\sum}  \mathbb{E}_{H}\left[\exp\left\{\frac{-{|H(\beta_i-\beta_j)|}^2}{4 \sigma^2 \mathrm{sin}^2\theta_t}\right\}\right],
		\end{aligned}
		\end{equation}
		where the first term satisfies \begin{equation}\label{eq8}
		\underset{\substack{\beta_i,\beta_j\in\Psi\\\beta_i=\beta_j}}{\sum}  \mathbb{E}_{H}\left[\exp\left\{\frac{-{|H(\beta_i-\beta_j)|}^2}{4 \sigma^2 \mathrm{sin}^2\theta_t}\right\}\right]=\underset{\substack{\beta_i,\beta_j\in\Psi\\\beta_i=\beta_j}}{\sum}1= 2^c,
		\end{equation} and the second term satisfies
		\begin{equation}
		\lim\limits_{\sigma\rightarrow 0}\underset{\substack{\beta_i,\beta_j\in\Psi\\\beta_i\ne\beta_j}}{\sum}  \mathbb{E}_{H}\left[\exp\left\{\frac{-{|H(\beta_i-\beta_j)|}^2}{4 \sigma^2 \mathrm{sin}^2\theta_t}\right\}\right]=0.
		\end{equation}
		As such, we have that $G\left(\Psi,0,\theta_t,f_H(h)\right)=2^c$. Applying it in \eqref{allinone} yields 
		\begin{equation}
		\mathscr{F} \left(L_a , 0\right)=\sum_{t=1}^N b_t \cdot 2^{-cL_a}=2^{-L_ac-1}.
		\end{equation}\vspace{-4mm}
	\end{proof} 
	
	\section{Proof of Theorem \ref{the3}}\label{appendixb}
	\begin{proof}
		For QAM modulation with even modulation order $c$, the average SNR can be calculated by
		\begin{equation}\label{eq12}
		\gamma=\frac{(2^c-1)\cdot d_{\min}^2}{6\sigma^2},
		\end{equation} 
		where $d_{\min}$ is the minimum distance of the QAM symbols. Substituting \eqref{eq12} into $G(\Psi,\sigma^2,\theta_t,f_H(h))$ yields \begin{equation}\label{eq13}
		\sum_{\beta_i,\beta_j \in \Psi}  \mathbb{E}_{H}\left[\exp\left\{\frac{-{3|H(\beta_i-\beta_j)|}^2\gamma}{2 (2^c-1)d_{\min}^2 \mathrm{sin}^2\theta_t}\right\}\right].
		\end{equation}
		On the other hand, the distance between a pair of QAM symbols forms a finite set, represented as $
		\mathcal{S}=\{|\beta_i-\beta_j|:\beta_i,\beta_j\in\Psi\}
		$. From Fig. \ref{fig:mapping}, we can numerate the above set as $\mathcal{S}=\{0,d_{\min},\sqrt{2}d_{\min},2d_{\min},\cdots,\sqrt{2}(2^{c/2}-1)d_{\min}\}$.
		With set $\mathcal{S}$, we can rewrite \eqref{eq13} as
		\begin{equation}\label{eq15}\begin{aligned}
		&\sum_{d\in\mathcal{S}}\underset{\substack{\beta_i,\beta_j\in\Psi\\|\beta_i-\beta_j|=d}}{\sum}  \mathbb{E}_{H}\left[\exp\left\{\frac{-{3|H(\beta_i-\beta_j)|}^2\gamma}{2 (2^c-1)d_{\min}^2 \mathrm{sin}^2\theta_t}\right\}\right]\\&=\sum_{d\in\mathcal{S}} A_d \mathbb{E}_{H}\left[\exp\left\{\frac{-{3|H|}^2d^2\gamma}{2 (2^c-1)d_{\min}^2 \mathrm{sin}^2\theta_t}\right\}\right],\\
		\end{aligned}
		\end{equation}
		where $A_d$ is the number of element pairs \((\beta_i, \beta_j)\) within \( \Psi \) that are exactly \( d \) units apart, with \begin{equation}
		A_d=|\{(\beta_i,\beta_j):|\beta_i-\beta_j|=d,\beta_i,\beta_j\in\Psi\}|.
		\end{equation}
		Denote $g(d,\gamma)=\mathbb{E}_{H}\left[\exp\left\{\frac{-{3|H|}^2d^2\gamma}{2 (2^c-1)d_{\min}^2 \mathrm{sin}^2\theta_t}\right\}\right]$, and \eqref{eq15} is reduced to $\sum_{d\in\mathcal{S}}A_d g(d,\gamma)$, which can be expanded by:
		\begin{equation}\label{eq17}	A_0+A_{d_{\min}}g(d_{\min},\gamma)+A_{\sqrt{2}d_{\min}}g(\sqrt{2}d_{\min},\gamma)\cdots.	\end{equation} 
		Because $g(nd,\gamma)=g(d,\gamma)^{n^2}$ and $g(d,\gamma)\ll1$ under high SNR $\gamma$, the tail of \eqref{eq17} is a higher-order infinitesimal relative to $g(d_{\min},\gamma)$. Therefore, \eqref{eq17} is simplified as
		\begin{equation}\label{eq18}
		A_0+A_{d_{\min}}g(d_{\min},\gamma)+o(g(d_{\min},\gamma)).	\end{equation} 
		From the proof of Lemma \ref{l1}, we know that 
		\begin{equation}
		\lim_{\gamma\rightarrow \infty} A_0+A_{d_{\min}}g(d_{\min},\gamma)+o(g(d_{\min},\gamma))=A_0,
		\end{equation}
		Thus, the SNR threshold can be approximated by $\gamma^{\text{th}}$, such that $A_{d_{\min}}g(d_{\min},\gamma^{\text{th}})\ll A_0$, or equivalently, \begin{equation}\label{eq21}
		\frac{A_{d_{\min}}g(d_{\min},\gamma^{\text{th}})}{A_0}\ll 1.
		\end{equation} 
		Note that $A_0=2^c$ and $A_{d_{\min}}\approx4\cdot 2^c$, substituting $g(d_{\min},\gamma^{\text{th}})=\mathbb{E}_{H}\left[\exp\left\{\frac{-{3|H|}^2\gamma^{\text{th}}}{2 (2^c-1) \mathrm{sin}^2\theta_t}\right\}\right]$ into \eqref{eq21} yields \begin{equation}\label{lhs}
		4\cdot \mathbb{E}_{H}\left[\exp\left\{\frac{-{3|H|}^2\gamma^{\text{th}}}{2 (2^c-1) \mathrm{sin}^2\theta_t}\right\}\right]\ll 1.
		\end{equation}
		Without loss of generality, setting $\theta = \pi/2$ allows us to impose a stricter condition, thus completing the proof.
	\end{proof}
\section{Proof of Corollary \ref{co1}}\label{appendixc}
\begin{proof}
	The expectation on the left-hand side of \eqref{the3eq} over different channels are:
	
	\noindent (i) \ul{\textit{Nakagami-m}}: \textcolor{black}{In this case, the Probability Density Function (PDF) of the modulus of $h$ is $f_{|H|}(r)=\frac{2 m^m}{\Gamma(m)} \cdot r^{2 m-1} \cdot \exp\{{-m r^2}\}$. Substituting it into the expectation yields:	}
	\begin{equation}\color{black}
    \begin{aligned}
        &\mathbb{E}_{H}\left[\exp\left\{\frac{-{3|H|}^2\gamma^{\text{th}}}{2 (2^c-1) }\right\}\right]=\mathbb{E}_{|H|}\left[\exp\left\{\frac{-{3|H|}^2\gamma^{\text{th}}}{2 (2^c-1) }\right\}\right]\\
            &=\int_0^\infty\exp\left\{\frac{-{3r}^2\gamma^{\text{th}}}{2 (2^c-1) }\right\}\cdot\frac{2 m^m r^{2 m-1}}{\Gamma(m)} \exp\left\{{-m r^2}\right\}\mathrm{d}r.
    \end{aligned}	
	\end{equation}
    \textcolor{black}{By implementing the substitution $q=\left(\frac{{3}\gamma^{\text{th}}}{2 (2^c-1) }+m\right)r^2$, the above improper integral turns to: }
    \begin{equation}\color{black}
    \begin{aligned}
           &\frac{m^m}{\left(\frac{{3}\gamma^{\text{th}}}{2 (2^c-1) }+m\right)^m\Gamma(m)}\underbrace{\int_0^\infty e^{-q}\cdot q^{m-1}}_{\Gamma(m)} \mathrm{d}q  \\
           &=        \left(\frac{2m(2^c-1)}{3\gamma^{\text{th}}+2m(2^c-1)}\right)^m.
    \end{aligned}
    \end{equation}
	(ii)\ul{\textit{ {Rayleigh}}}: \textcolor{black}{In this case, the PDF of the modulus of $h_{i,j}$ is $f_{|H|}(r)=2r\cdot\exp\{-r^2\}$. We have:}
	\begin{equation}\color{black}
    \begin{aligned}
    	&\mathbb{E}_{H}\left[\exp\left\{\frac{-{3|H|}^2\gamma^{\text{th}}}{2 (2^c-1) }\right\}\right]=\mathbb{E}_{|H|}\left[\exp\left\{\frac{-{3|H|}^2\gamma^{\text{th}}}{2 (2^c-1) }\right\}\right]\\&=\int_0^\infty\exp\left\{\frac{-{3r}^2\gamma^{\text{th}}}{2 (2^c-1) }\right\}\cdot{2r}\cdot\exp\{-r^2\}~\mathrm{d}r.
    \end{aligned}
	\end{equation}
    \textcolor{black}{By implementing the substitution $q=\left(\frac{{3}\gamma^{\text{th}}}{2 (2^c-1) }+1\right)r^2$, the above improper integral turns to: }
    \begin{equation}\color{black}
        \frac{1}{\frac{{3}\gamma^{\text{th}}}{2 (2^c-1) }+1}\int_{0}^\infty e^{-q}\mathrm{d}q=\frac{2(2^c-1)}{3\gamma^{\text{th}}+2(2^c-1)}.
    \end{equation}
(iii) \ul{\textit{AWGN}}: \textcolor{black}{In this case, $|H|\equiv1$ holds, and we have:	}
	\begin{equation}
    \begin{aligned}
        \mathbb{E}_{H}\left[\exp\left\{\frac{-{3|H|}^2\gamma^{\text{th}}}{2 (2^c-1)}\right\}\right]=\exp\left\{\frac{-{3}\gamma^{\text{th}}}{2 (2^c-1)} \right\}.
    \end{aligned}
	\end{equation}
	Then, by setting the left-hand side of \eqref{the3eq} as $x$ and explicitly solving $\gamma^{\text{th}}$, we can obtain the threshold $\gamma^{\text{th}}$ over the \textcolor{black}{nagakami-m}, \textcolor{black}{Rayleigh}, and AWGN channel, respectively.
\end{proof}
	\bibliographystyle{IEEEtran}
	\bibliography{reference}

\end{document}